\documentclass[10pt,a4paper,twocolumn]{IEEEtran}

\usepackage{amsopn,amsmath,amssymb,amsfonts,bm,amsthm}
\usepackage{graphicx,subfigure}
\usepackage{enumerate}
\usepackage{cite}
\usepackage{flushend}
\usepackage{etex}
\usepackage{algorithm,algorithmic}
\usepackage{tikz,pgfplots}
\usepackage[margin=16mm,top=20mm,bottom=16mm]{geometry}
\newcommand{\figref}[1]{Fig.\,\ref{#1}}

\def\a{\alpha}
\def\mA{\boldsymbol A}
\def\Dmin{D_{\min}}
\def\Dr{D_\mathrm{RLNC}}
\def\Du{\underline{D}}
\def\E{\mathcal E}
\def\e{\boldsymbol e}

\def\F{\mathbb F}

\def\G{\mathcal G}
\def\mH{\mathcal H}
\def\M{\mathcal M}

\def\P{\mathcal{P}}
\def\p{\mathbf{p}}
\def\R{\mathcal{R}}
\def\r{\mathbf r}
\def\S{\mathcal S}
\def\w{\mathcal W}

\def\V{\mathcal V}
\def\v{\boldsymbol v}

\newtheorem{Lemma}{\textbf{Lemma}}
\newtheorem{Theorem}{\textbf{Theorem}}

\newtheorem{Definition}{\textbf{Definition}}

\newtheorem{Conjecture}{Conjecture}

  {\proof}{\proofend}

\title{\huge{On Minimizing the Average Packet Decoding Delay in Wireless Network Coded Broadcast}}
\author{Mingchao Yu*, Alex Sprintson$^\dag$, and Parastoo Sadeghi* \\ \small{*Research School of Engineering, Australian National University, Canberra, Australia\\$^\dag$Department of Electrical and Computer Engineering, Texas A\&M University, Texas, USA\\ \texttt{Emails: \{ming.yu,parastoo.sadeghi\}@anu.edu.au,~spalex@tamu.edu}}\\
\vspace{-2em}}
\begin{document}
\maketitle
\thispagestyle{empty}
\vspace{-5em}
\begin{abstract}
We consider a setting in which a sender wishes to broadcast a block of $K$ data packets to a set of wireless receivers, where each of the receivers has a subset of the data packets already available to it (e.g., from prior transmissions) and wants the rest of the packets. Our goal  is to find a linear network coding scheme that yields the minimum average packet decoding delay (APDD), i.e., the average time it takes for a receiver to decode a data packet. Our contributions can be summarized as follows. First, we prove that this problem is NP-hard by presenting a reduction from the hypergraph coloring problem. Next, we show that 
a random linear network coding (RLNC) provides an approximate solution to this problem with approximation ratio $2$ with high probability. Next, we present a methodology for designing specialized approximation algorithms for this problem that outperform RLNC solutions while maintaining the same throughput. In a special case of practical interest with a small number of wanted packets our solution can achieve an approximation ratio $\frac{4-2/K}{3}$. Finally, we conduct an experimental study that demonstrates the advantages of the presented methodology.
	


\end{abstract}

\section{Introduction}

In this paper, we are interested in a wireless broadcast scenario, in which a sender wishes to broadcast a block of $K$ data packets to a set of wireless receivers, such that each of the receivers already has a subset of the data packets available to it (e.g., from prior transmissions) and is interested in obtaining the rest of the packets. Given a packet reception instance, the goal is to design a linear network coding (NC) scheme that minimizes the average packet decoding delay (APDD), which is defined as the average time it takes for a receiver to decode a data packet.

One of the possible solutions to this problem is to employ
a random linear  network coding (RLNC) technique \cite{ho:medard:koetter:karger:effros:2006,nistor:lucani:vinhoza:costa:barros:2011}. In wireless broadcast scenarios, RLNC can achieve an optimal throughput (i.e., minimize the time required to decode all packets by all receivers) with high probability by mixing all data packets in the block together using linear coefficients randomly chosen from a sufficiently large finite field. However, RLNC is suboptimal in terms of APDD, since in general, no data packet can be decoded by a receiver until it receives $K$ linearly independent coded packets.

%

Many opportunistic NC techniques have been developed with the aim to reduce APDD or some other measures of decoding delay \cite{katti1;etal:2008,Rozner_Heuristic_clique,costa:munaretto:widmer:baros:2008,keller:drinea:fragouli:2008,
eryilmaz:ozdaglar:medard:ahmed:2008,barros:costa:munaretto:widmer:2009,sundararajan:sadeghi:medard:2009,
nguyen:tran:nguyen:bose:2009,sadeghi:shams:traskov:2010,sorour:valaee:2010,nistor:lucani:vinhoza:costa:barros:2011,
li:idnc_video:2011,athanasiadou2013stable,yu:parastoo:neda:2014,parastoo:yu:neda:isita,fu2014dynamic}.  An important technique in this class is \emph{instantly decodable network coding} (IDNC). The IDNC technique has a potential to reduce the APDD by enabling a subset of receivers to instantly decode a data packet after each  transmission. IDNC has been shown to outperform RLNC in terms of APDD for a small number of receivers \cite{yu:parastoo:neda:2014}.  However, since in IDNC schemes a single transmission typically benefits only some of the receivers, IDNC is not throughput optimal. As a result, a larger number of transmissions is necessary to finish the broadcast, which increases the decoding delay for some receivers, and, as a result, increases the value of APDD. Indeed, for larger number of receivers, the throughput of IDNC decreases and APDD increases due to lack of coding opportunities.  A similar behavior can be observed for other opportunistic coding techniques \cite{costa:munaretto:widmer:baros:2008}.

In summary, there is no clear winner between RLNC, IDNC, and other opportunistic techniques, as each of them prevail in a different  parameter region. Moreover, while the APDD of RLNC can be easily calculated (as shown in Section~\ref{sec:approx}), the achievable APDD of opportunistic NC techniques has not been characterized analytically.



The contributions of this paper is summarized as follows:
\begin{itemize}
\item We first prove that it is NP-hard to minimize APDD, by presenting a reduction from the hypergraph coloring problem. 
\item Next we show that RLNC achieves an approximation ratio of 2 with high probability, i.e., the APDD achieved by RLNC is at most two times the optimal solution.
\item We present a methodology for designing specialized approximation algorithms that achieve lower values of APDD than RLNC while maintaining the same optimal throughput. We also present a case study to demonstrate the algorithm design. We conduct extensive simulations to confirm that our methodology outperforms alternative solutions in the broad range of practical settings. 


\end{itemize}


\section{System Model}
Our model includes a single sender that holds a set of $K$ data packets that belong to $\F_q$, $\P=\{\p_k\}_{k=1}^K$, and a set of $N$ receivers, $\{\r_n\}_{n=1}^N$, each wants a subset $\w_n$ of $\P$ and has the rest.
The packet reception instance in our model is represented by a binary $N\times K$ state feedback matrix (SFM) $\mA$, where $\mA(n,k)=1$ means that $\r_n$ wants $\p_k$, and $\mA(n,k)=0$ means that receiver $r_n$ has packet $p_k$ already available to it (e.g., from prior transmissions). We denote by $w_n$ the size of $\w_n$, and by $t_k$ the number of receivers who want $\p_k$. An example of SFM is given in \figref{fig:sfm_graph}(b), which has $w_1=3$ and $t_1=2$.

Given $\mA$, the sender performs a linear NC transmission phase. In each NC transmission, the sender encodes data packets in $\M$ together using linear coefficients from a finite field $\F_q$. The corresponding packet $X$ takes the form of:
\begin{equation}
X=\sum_{\p_k\in\M}\beta_k\p_k
\end{equation}
We denote by $\M=\{p_k\in \P| \beta_k\neq 0\}$ the support of $X$ and refer to it as an \emph{coding set} of $X$.
When the coefficients $\{\beta_k\}$ are chosen from $\F_q$ uniformly at random, $X$ is called a random-coded packet of $\M$. A receiver will increase its degree of freedom (DoF) by one when it receives a NC packet that is linearly independent of the set of all the packets it already has. The broadcast will be \emph{completed} at a receiver once it decodes all its wanted data packets.

In order to study the global minimum decoding delay of linear NC in wireless broadcast, we assume the following:
\begin{enumerate}
\item NC transmissions are erasure-free, so that every transmitted NC packet can be received by all receivers;
\item Receivers have sufficient computational resources to perform NC decoding under any $\F_q$. When random coding is applied, a sufficiently large $\F_q$ will ensure the linear independency among the random-coded packets with high probability.
\end{enumerate}

A set of $U$ coded packets is called a \emph{NC solution} and is denoted by $\S$ if it allows
every receiver to decode all its wanted packets. Let $u_{n,k}$ be the index of the NC transmission at which $\r_n$ decodes $\p_k$. The average packet decoding delay (APDD) of $\S$, denoted by $D_\S$, is calculated as:
\begin{equation}
D_\S=\frac{1}{\sum_{n=1}^Nw_n}\sum_{n=1}^N\sum_{k=1}^Ku_{n,k}.
\end{equation}

Our aim in this paper is to study the smallest $D_\S$ over all possible linear NC solutions. We call it the minimum APDD of $\mA$ and denote it by $\Dmin$. The first question we would like to answer is: \emph{Is it hard to find $\Dmin$?}

\begin{figure}
\centering
\subfigure[Hypergraph $\mH$]{\includegraphics[width=0.3\linewidth]{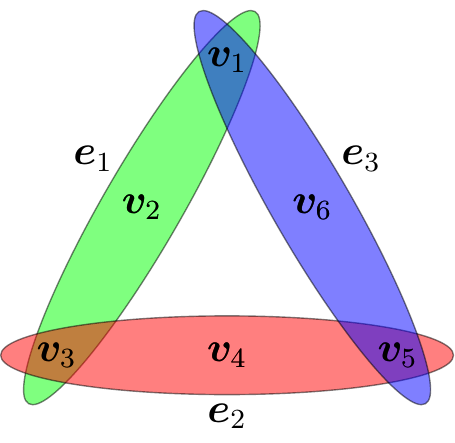}}\hspace{20pt}
\subfigure[State feedback matrix $\mA$]{\includegraphics[width=0.55\linewidth]{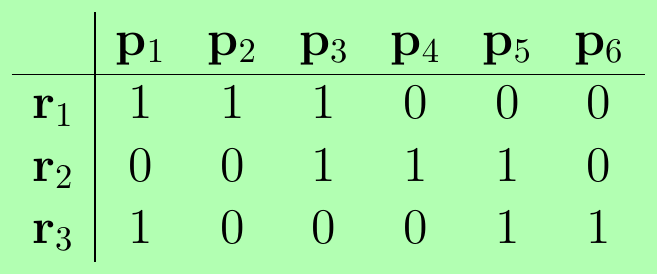}}
\caption{A hypergraph and its equivalent SFM $\mA$}
\label{fig:sfm_graph}
\end{figure}

\vspace{-0.3em}
\section{The Hardness of Finding $\Dmin$}
In this section, we study the hardness of finding $\Dmin$. To this end, we will first introduce the concept of \emph{perfect NC solution}, whose APDD is a lower bound of $\Dmin$. Then, we will prove that deciding whether a perfect solution exists for a given instance $\mA$ of the problem at hand is an NP-hard problem. This implies, in turn, that the problem of finding $\Dmin$ is also NP-hard.

\subsection{The Perfect Solution and a Bound of $\Dmin$}
\begin{Definition}
A NC solution $\S$ is called a perfect solution and is denoted by $\S_p$ if it allows every receiver $n$ to decode a wanted data packet in every transmission of $\S$.
\end{Definition}

Obviously, $\S_p$ offers the ideal packet decoding scenario. Its APDD is thus a lower bound of $\Dmin$, denoted by $\underline D$, which is calculated as:
\begin{align}
\Du&=\frac{1}{\sum_{n=1}^Nw_n}\sum_{n=1}^N\sum_{u=1}^{w_n} u\\
             &=\frac{1}{\sum_{n=1}^Nw_n}\sum_{n=1}^N\frac{(w_n+1)w_n}{2}\\
             &=\frac{\sum_{n=1}^Nw_n^2}{2\sum_{n=1}^Nw_n}+\frac{1}{2}\label{eq:du}
\end{align}
It is clear that $\Du$ can only be achieved by $\S_p$ if a perfect solution exists. The natural question =in this context is: \emph{Does a perfect solution $\S_p$ exist for every SFM?} In the next subsection, by using a reduction from the strong hypergraph coloring problem, we will prove that this question is NP-hard to answer.

\subsection{Hardness of Finding $\S_p$}
We first introduce some useful concepts in hypergraphs. A hypergraph $\mH$ is defined by a pair $(\V,\E)$, where $\V$ is the set of vertices, and $\E$ is the set of hyperedges. Every hyperedge $\e\in\E$ is a subset of $\V$ with size $|\e|\geqslant 1$. A hypergraph is $r-$uniform if every hyperedge $\e$ has equal size, i.e., $|\e|=r$. A $k$-\emph{strong} coloring solution of $\mH$ is a partition of $\V$ into $k$ subsets $\{\V_i\}_{i=1}^k$, such that $|\V_i\cap\e|\leqslant 1$ for any $\e\in\E$. In other words, every color appears at most once in every hyperedge.
It is well known that the hypergraph coloring problem is intractable.

\begin{Lemma}[\hspace{-0.5pt}\cite{agnarsson2005strong}]\label{lemma:strong_color}
It is NP-hard to determine whether an $r-$uniform hypergraph is $r$-strong colorable, for any $r\geqslant3$.
\end{Lemma}

We then build a reduction from the strong hypergraph coloring solution for $r$-uniform hypergraphs to the problem of finding a perfect NC solution for the average delay minimization problem. Given an $r$-uniform hypergraph $\mH(\V,\E)$ we construct an instance to our problem as follows. First, for each vertex $\v_k$ we introduce a data packet $\p_k$, and for each hyperedge $\e_n$ we introduce a receiver $\r_n$ who wants the data packets that correspond to vertices in $\e_n$. 
Note that in the resulting SFM $\mA$, every receiver wants $r$ data packets. A 3-uniform hypergraph and the corresponding SFM matrix are depicted in \figref{fig:sfm_graph}.


First, we prove that an existence of an $r$-strong coloring solution $\{\V_i\}_{i=1}^r$ of $\mH$ implies a perfect solution for our problem. Let $\{\V_i\}_{i=1}^r$ be an $r$-strong coloring of $\mH$. For each $\V_i$, let $\M_i$  be a set of packets that correspond to vertices in $\V_i$. Note that for each each receiver $\r_n$ and each set $\M_i$ it holds that $|\M_i\cap\w_n|=1$. Consider a coding solution $\S$ that includes $r$ transmissions, such that transmission $i$ includes a sum of packets in $\M_i$ (over $\F_q$). Since every receiver can decode a packet at each transmission, $\S$ is a perfect solution to our our problem.


Next, we show that a perfect solution $\S_p$ for the instance of $\mA$ of our problem implies that there exists an  $r$-strong coloring solution $\{\V_i\}_{i=1}^r$ of $\mH$. Let $\M_i$ be the coding set that corresponds to the transmission $i$ of $\S_p$ and let $V_i$ be the set of vertices in $\mH$ that correspond to $\M_i$. Note that in order to allow every receiver to decode one data packet in each of the $r$ transmissions, every $\M_i$ must contain one wanted data packet of every receiver, i.e., $|\M_i\cap\w_n|=1$. Thus, $\{\V_i\}_{i=1}^r$ is an $r$-strong coloring solution $\{\V_i\}_{i=1}^r$ of $\mH$.


We conclude that an $r$-uniform hypergraph is $r$-strong colorable if and only if there exists a perfect NC solution of the instance $\mA$ of our problem. We summarize our results in the following lemma:

\begin{Lemma}\label{cor:det_perfect}
It is NP-hard to determine whether there exists a perfect solution for a given instance $\mA$ of minimum APDD problem.
\end{Lemma}
\begin{proof}
The theorem follows from our construction and  Lemma~\ref{lemma:strong_color}.
\end{proof}

\subsection{The Hardness of Finding $\Dmin$}

Since $\underline{D}$ can only be achieved by a perfect solution $\S_p$, an optimal algorithm that finds $\Dmin$ will be able to determine the existence of a perfect solution by comparing $\Dmin$ with $\underline{D}$. According to Lemma~\ref{cor:det_perfect}, this decision is NP-hard to made, and thus it is NP-hard to find $\Dmin$:

\begin{Theorem}
It is NP-hard to find $\Dmin$ for a given instance $\mA$ of minimum APDD problem .
\end{Theorem}

In addition to NP-hardness, our reduction from the  hypergraph coloring also yields an interesting conjecture on the existence of perfect solution for some special instances of $\mA$. It comes from the famous Erd\H{o}s-Faber-Lov\'asz conjecture in graph theory \cite{erdos1981combinatorial}:
\begin{Conjecture}[Erd\H{o}s-Faber-Lov\'asz \cite{erdos1981combinatorial}]
Consider an \mbox{$r$-uniform hypergraph} with $r$ hyperedges. Each pair of hyperedges have at most one vertex in common. This hypergraph is $r$-strong colorable.
\end{Conjecture}
The corresponding conjecture in NC context is as follows:
\begin{Conjecture}
Consider an instance $\mA$ of our problem with $r$ receivers, each wants $r$ data packets, and each pair of receivers want at most one data packet in common. This $\mA$ has a perfect NC solution $\S_p$.
\end{Conjecture}

We showed that the problem of finding a minimum value of $\Dmin$ is intractable. Accordingly, in the next sections we discuss approximation algorithms for this problem.


\section{Approximating $\Dmin$ }\label{sec:approx}
In this section, we aim at approximating $\Dmin$. An approximation algorithm of $\Dmin$ produces a linear NC solution $\S$ with its APDD obeying $D_\S\leqslant \a\Dmin$. We refer to $\a\geqslant 1$ as \emph{approximation ratio} of the algorithm.


In the next theorem we analyze the approximation ratio of the RLNC technique:

\begin{Theorem}
RLNC technique is at most a $2-$approximation algorithm of $\Dmin$.
\end{Theorem}

\begin{proof}
In every RLNC transmission, the sender sends a random-coded packet of all data packets. With high probability (that asymptotically goes to 1 with the field size), after receiving $w_n$ such packets, receiver $\r_n$ can decode all its wanted data packets by performing block decoding, i.e., solving a set of $w_n$ linear equations. Hence, the APDD offered by RLNC is:
\begin{align}
\Dr=\frac{1}{\sum_{n=1}^Nw_n}\sum_{n=1}^Nw_n^2.
\end{align}
Comparing $\Dr$ with the lower bound $\underline D$ in \eqref{eq:du}, we have:
\begin{equation}
\frac{\Dr}{\Du}=\frac{\frac{\sum_{n=1}^Nw_n^2}{\sum_{n=1}^Nw_n}}{\frac{\sum_{n=1}^Nw_n^2}{2\sum_{n=1}^Nw_n}+\frac{1}{2}}< 2.
\end{equation}
Since $\Dmin\geqslant \Du$, $\Dr$ at most doubles $\Dmin$. Thus, RLNC is at most a 2-approximation algorithm of $\Dmin$.
\end{proof}

Therefore, RLNC technique offers guaranteed APDD performance. On the other hand, to the best of our knowledge existing opportunistic APDD-reduction techniques are not able to provide provable performance guarantees.  For example, let us analyze a well-known APDD-reduction technique called instantly decodable network coding (IDNC).

IDNC has two variations, strict IDNC (S-IDNC) \cite{Rozner_Heuristic_clique,sundararajan:sadeghi:medard:2009,yu:parastoo:neda:2014} and general IDNC (G-IDNC) \cite{sorour:valaee:2010}. Both of them have been shown to provide lower APDD than RLNC with a small number of receivers, but become worse than RLNC with increasing number of receivers. Due to the absence of the optimal G-IDNC algorithm \cite{sorour:valaee:2010}, we are not able to prove whether G-IDNC approximates $\Dmin$ or not. However, we are able to prove the following statement for S-IDNC:
\begin{Lemma}\label{theo:idnc}
S-IDNC does not provide a constant approximation ratio for the minimum APDD problem.
\end{Lemma}
\begin{proof}
To prove this, it suffices to provide a counter example. Consider a complete graph $\G(\V,\E)$ with $K$ vertices and $K(K-1)/2$ edges. For every vertex $\v_k$ we generate a data packet $\p_k$. For every edge $\e_{i,j}$ that connects $\v_i$ and $\v_j$ we generate a receiver $\r_n$ with $\w_n=\{\p_i,\p_j\}$. In the resultant $\mA$, every receiver wants two data packets.

S-IDNC prohibits to code together any two data packets that are both wanted by any receiver. In other words, every S-IDNC coding set $\M$ must satisfy $|\M\cap\w_n|\leqslant 1$ for any receiver $\r_n$. Given the above $\mA$, this restriction implies that no data packets can be coded together at all. Hence, all $K$ data packets must be broadcast uncoded alone. The resultant APDD is  $(K+1)/2$. 

Note that it is easy to show that the optimal value of APDD is at most 2. Indeed, the value of 2 can be achieved by using the RLNC technique. Thus, S-IDNC fails to provide a constant approximation ratio for the problem at hand.

\end{proof}

In conclusion, in this section we proved that RLNC is at least a $2$-approximation algorithm of $\Dmin$. By setting RLNC as a benchmark, we showed that S-IDNC fails to provide a constant approximation ratio for our problem. Indeed, RLNC is the only existing approximation algorithm, to the best of our knowledge. Therefore, the final question we are interested in is: \emph{How to overtake RLNC}?

\section{How to Overtake RLNC}
Imagine a linear NC technique that: 1) is throughput optimal as RLNC (i.e., allows every receiver to increase its DoF by one in every transmission \cite{keller:drinea:fragouli:2008}); and 2) enables early packet decodings rather than block decodings in RLNC. Such an NC technique offers an APDD lower than RLNC, and thus will be an approximation algorithm of $\Dmin$ with a  ratio lower than RLNC. To the best of our knowledge, such NC techniques have not been developed in the literature.

In this section, we propose a methodology for the development of such NC techniques. We first construct a hypergraph $\mH$ that corresponds to a given instance $\mA$. The key idea to guarantee optimal throughput and early packet decodings is to find minimal vertex covers in $\mH$.
%
%
A vertex cover is a subset $\V_C$ of $\V$ such that $|\V_C\cap\e_n|\geqslant 1$ for every hyperedge. It is minimal if it is not the superset of a smaller vertex cover, implying that $|\V_C\cap\e_n|=1$ for at least one hyperedge. Hence, every receiver wants at least one data packet from $\V_C$, and at least one receiver can instantly decode a wanted data packet from $\V_C$.

\begin{algorithm}[t]
\caption{Structure of approximation algorithms of $\Dmin$}
\label{alg:partition}
\begin{algorithmic}[1]
\STATE Initialize: SFM $\mA$;
\STATE Construct the hypergraph $\mH(\V,\E)$ of $\mA$ by mapping data packets to vertices and receivers to hyperedges;
\WHILE {Every hyperedge $\e_n$ is non-empty}
\STATE Find a minimal vertex cover $\V_C$ of $\mH$, and send a random-coded packet of data packets in $\V_C$;
\STATE Update $\mH$ by removing $\V_C$ from it following some strategy;
\ENDWHILE
\STATE Send random-coded packets of all data packets until all receivers complete broadcast.
\end{algorithmic}\label{alg:dmin}
\end{algorithm}

The core algorithmic structure of our methodology is sketched in Algorithm \ref{alg:dmin}. It generates a solution $\S$ with coding sets $\{\V_C^1,\cdots,\V_C^L,\P,\P,\cdots\}$, where $L$ is the total number of minimal vertex covers found by the algorithm. To achieve optimal throughput, the random-coded packet of every $\{\V_C^l\}_{l=1}^L$ must be able to increase every receiver's DoF by one. To this end, a proper hypergraph update strategy must be applied. The simplest strategy is to completely remove $\V_C^l$ from $\mH$ before finding $\V_C^{l+1}$. By doing so, all the vertex covers will have empty intersections, and thus serve all the receivers with different data packets. The algorithm stops at the $L$-th round when there is at least one empty hyperedge, after which point, optimal throughput is maintained by sending random-coded packets of all data packets in $\P$, as in RLNC. Hence, the solution $\S$ is throughput optimal as RLNC. Moreover, since the minimal vertex covers enable instant packet decodings, the APDD of $\S$ is better than RLNC.


The design of optimal hypergraph vertex cover algorithms and hypergraph update strategies that minimize APDD is still an open problem. However, regardless of whether optimal or heuristic algorithms/strategies are applied, solutions generated by Algorithm 1 are always throughput optimal, while also providing early packet decodings. Thus, they can approximate $\Dmin$ with ratio smaller than RLNC.

\figref{fig:d_general} compares the APDD performance of a simple realization of our methodology with RLNC and a heuristic G-IDNC \cite{sameh:valaee:globecom:2010}. In this realization, we adopt the aforementioned complete $\V_C$ removal strategy and a heuristic hypergraph vertex cover algorithm, which iteratively adds to $\V_C$ the vertex that 1) is not connected to $\V_C$; and 2) has the highest degree\footnote{The degree of a vertex is the number of hyperedges incident to it}. To generate the SFM instances, we consider a block of $K=15$ data packets and assume that each receiver wants each data packet randomly with a probability of $0.2$. The number of receivers $N\in[5,100]$. The results show that G-IDNC offers the lowest APDD when $N$ is small, but becomes worse than RLNC when $N>65$. Hence, the heuristic G-IDNC is not an approximation algorithm. Our realization (heur. VC in the figure) always outperforms RLNC. Their gap narrows down with increasing $N$.

To gain a deeper insight into the realization and performance analysis of the proposed methodology, we conduct a case study in the next subsection by considering a special type of SFM.

\begin{figure}[t]
\centering
\includegraphics[width=\linewidth]{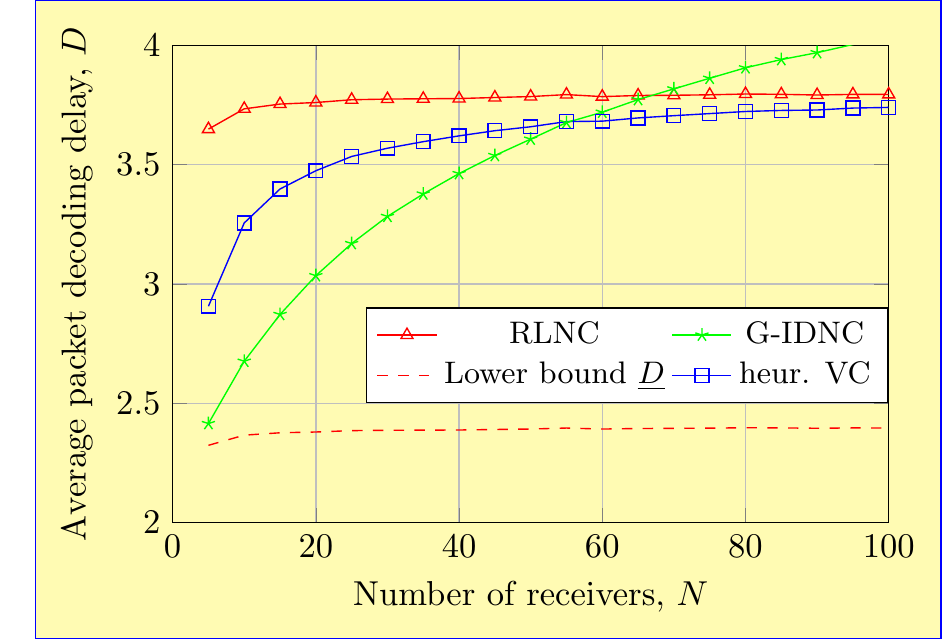}
\caption{The APDD of heuristic VC, RLNC, and G-IDNC when $K=20$.}
\label{fig:d_general}
\end{figure}
\vspace{-2mm}
\subsection*{A Case Study}
In this subsection, we design a NC technique that approximates the $\Dmin$ of a special type of SFM where every receiver wants two data packets from the packet block $\P$. Though seemingly simple, this case is highly nontrivial, because its $\Dmin$ cannot be approximated by existing APDD-reduction techniques such as S-IDNC, as shown in the proof of Theorem \ref{theo:idnc}.

Given such an SFM $\mA$, we first construct its hypergraph model by mapping data packets into vertices, and mapping receivers into hyperedges. Note that multiple receivers who want the same set of data packets are represented by one hyperedge. In addition, we weight every vertex $\v_k$ with a value of $t_k$, which is the number of receivers who want $\p_k$. We note that since every hyperedge has $|\e|=2$, the resultant hypergraph is indeed a classic graph $\G(\V,\E)$.

We then partition $\V$ into two subsets:
\begin{itemize}
\item The first subset is a minimal vertex cover $\V_C$. Since every edge is incident to $\V_C$, every receiver wants at least one data packet from $\V_C$. Denote by $\R_C$ the receivers who want two data packets from $\V_C$, and by $N_C$ their number;
\item The second subset is a set $\V_I=\V\setminus\V_C$. It is obvious that $\V_I$ is a maximal independent set, because it contains no edge (otherwise $\V_C$ is not a minimal vertex cover). Hence, every receiver wants at most one data packet from $\V_I$. We denote by $\R_I$ the set of receivers who want one data packet from $\V_I$, and by $N_I$ their number. We have $N_C+N_I=N$.
\end{itemize}
An example of such partition is demonstrated in \figref{fig:mis_example}. It has 4 data packets and 5 receivers, with 4 of them want one data packet from both $\V_I$ and $\V_C$.
\begin{figure}
\centering
\includegraphics[width=0.3\linewidth]{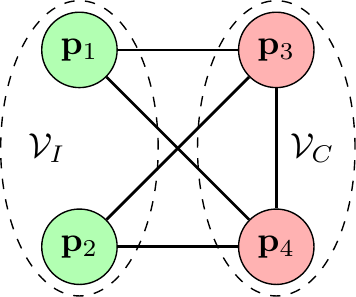}
\caption{An example of MIS partition.}
\label{fig:mis_example}
\end{figure}

We then send the following two NC packets:
\begin{itemize}
\item In the first transmission, send an RLNC packet of all data packets in $\V_C$. This allows $\R_I$ to decode one data packet, and allows $\R_C$ to increase DoF by one without decoding;
\item In the second transmission, send an RLNC packet of all data packets in $\V$. This allows $\R_I$ to decode the other data packet, and allows $\R_C$ to decode two data packets.
\end{itemize}
We call this technique maximal independent set (MIS) technique. MIS follows Algorithm \ref{alg:dmin}, and thus is an approximation algorithm of $\Dmin$. Its APDD, denoted by $D_{\mathrm{MIS}}$, is calculated as:
\begin{equation}
D_{\mathrm{MIS}}=\frac{N_I \cdot 1+ N_I\cdot2 +2N_C\cdot2}{2N}=2-\frac{N_I}{2N}
\end{equation}
which is minimized when $N_I$ is maximized. Since $N_I=\sum_{\v_k\in\V_I}t_k$, we need to find the maximum weighted independent set $\V_I$, which is NP-hard \cite{Graph_theory}. Nevertheless, even a heuristically finding $\V_I$ can offer $N_I>0$. Hence, $D_{\mathrm{MIS}}<2$ regardless of the way $\V_I$ is found.

We now derive the worst approximation rate of MIS by calculating an upper bound on $D_{\mathrm{MIS}}$. The minimum size of $\V_I$ is one, taking place when $\G$ is complete. In this case, the optimal MIS will find the solo vertex with the largest weight. Hence, $N_I/N$ is minimized when all vertices have the same weight. In this case, we have $\frac{N_I}{N}=\frac{2}{K}$. $D_{\mathrm{MIS}}$ is thus upper bounded as $D_{\mathrm{MIS}}\leqslant2-\frac{1}{K}$.  Then, by noting that $\Du=1.5$ when every receiver wants 2 data packets, we conclude that:
\begin{Theorem}
MIS is at most a $\frac{4-2/K}{3}$-approximation algorithm of $\Dmin$ when every receiver wants 2 out of $K$ packets.
\end{Theorem}

\begin{figure}[t]
\centering
\includegraphics[width=\linewidth]{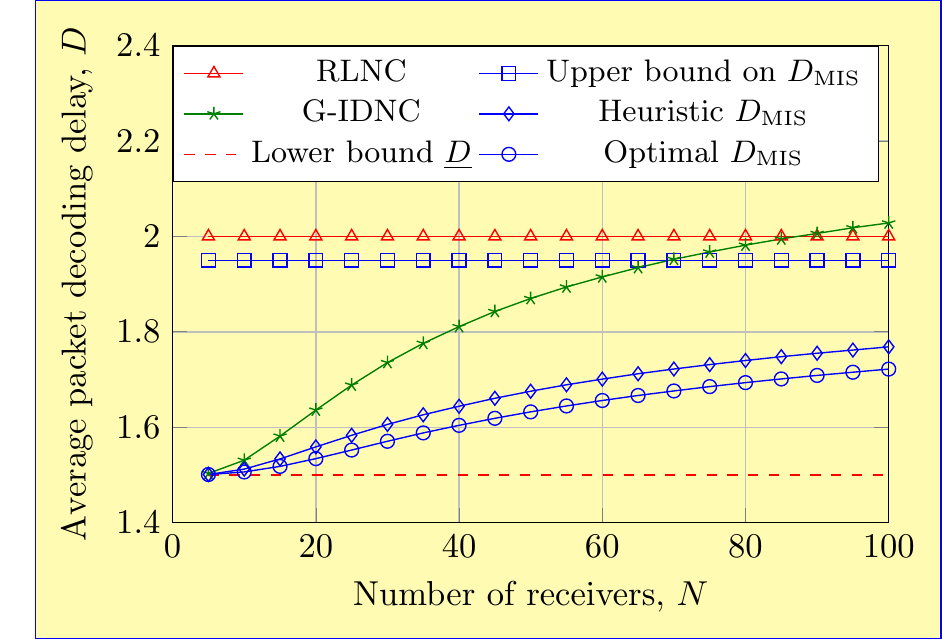}
\caption{The APDD of the proposed technique, RLNC, and G-IDNC when $K=20$ and $w=2$.}
\label{fig:d_compare}
\end{figure}
\figref{fig:d_compare} compares the simulated APDD performance of MIS with RLNC and the heuristic G-IDNC \cite{sameh:valaee:globecom:2010}. The packet block size is $K=20$. The number of receivers $N\in[5,100]$. Every receiver randomly chooses two wanted data packets. Since $K=20$, $D_{\mathrm{MIS}}$ is upper bounded by $2-\frac{1}{K}=1.95$. The optimal $D_{\mathrm{MIS}}$ is obtained by exhaustively searching the maximum weighted independent set. Both the performance of heuristic $D_{\mathrm{MIS}}$ and heuristic G-IDNC are obtained by using the heuristic maximum weighted clique (a complete subgraph of a graph) search algorithm proposed in \cite{sameh:valaee:globecom:2010}. This algorithm can be adapted for MIS because an independent set of $\G$ is a clique of the complementary graph $\overline\G$. According to the results, both the optimal and heuristic $D_{\mathrm{MIS}}$ are well below their upper bound, and are much better than both G-IDNC and RLNC. On the other hand, the APDD of G-IDNC exceeds RLNC when the number of receivers becomes large.


\section{Conclusion}
In this paper, we proved that it is NP-hard to minimize the average packet decoding delay (APDD) in packet block based wireless broadcast using linear network coding. But the minimum APDD can be approximated by RLNC with a ratio of at most 2. In order to achieve a lower approximation rario, we proposed a methodology for the design of specialized approximation algorithms that always outperform RLNC.

In the future, we are interested in designing more sophisticated realizations of the proposed NC framework. We are also interested in its extension to more general network settings, for example, when NC transmissions are subject to erasures. Besides, our hypergraph model and delay analysis may be extended to other network models such as cooperative data exchange and distributed data storage, because they also have similar types of demands on data packets.

\bibliographystyle{IEEEtran}

\end{document}